\documentclass[12pt]{article}
\usepackage{amssymb,amsmath,latexsym}
\usepackage{graphicx}
\usepackage[final]{pdfpages}
\usepackage{verbatim}
\usepackage{color}
\usepackage{amsfonts,amsthm,hyperref,mathrsfs,ulem,tikz}

\usetikzlibrary{arrows,patterns}

\title{Band edge limit of the scattering matrix \\ for quasi-one-dimensional discrete Schr\"odinger operators}

\author{Miguel Ballesteros$^1$, Gerardo Franco $^1$, \\ Guillermo Garro$^1$,   Hermann Schulz-Baldes$^2$
\\
\\
{\small $^1$ IIMAS, UNAM, Mexico}
\\   
{\small $^2$Department Mathematik, Friedrich-Alexander-Universit\"at Erlangen-N\"urnberg, Germany}
}

\date{ }

\newtheorem{theo}{Theorem}
\newtheorem{defini}[theo]{Definition}
\newtheorem{proposi}[theo]{Proposition}
\newtheorem{lemma}[theo]{Lemma}

\newtheorem{rem}[theo]{Remark}

\newtheorem{as}[theo]{Assumption}

\newcommand{\CM}{{\mathbb C}}
\newcommand{\NM}{{\mathbb N}}
\newcommand{\RM}{{\mathbb R}}
\newcommand{\SM}{{\mathbb S}}

\newcommand{\ZM}{{\mathbb Z}}

\newcommand{\DM}{{\mathbb D}}

\newcommand{\Ss}{{\cal S}}

\newcommand{\Nn}{{\cal N}}
\newcommand{\Mm}{{\cal M}}

\newcommand{\Ll}{{\cal L}}

\newcommand{\Ker}{\mbox{\rm Ker}}

\newcommand{\one}{{\bf 1}}

\newcommand{\overz}{1/z}
\newcommand{\overzb}{1/\overline{z}}
\newcommand{\overzz}{z^{-1}}

\newcommand{\zb}{\overline{z}}

\newcommand{\bsm}{\left(\begin{smallmatrix}} 
\newcommand{\esm}{\end{smallmatrix}\right)}  
\definecolor{GR}{rgb}{.35,.7,.35}

\usepackage{cancel}

\begin{document}

\maketitle

\begin{abstract}

 This paper is about the scattering theory for one-dimensional matrix Schr\"odinger operators with a matrix potential having a finite first moment. The transmission coefficients are analytically continued and extended to the band edges. An explicit expression is given for these extensions. {The} limits  of the reflection coefficients at the band edges is also calculated.

\vspace{.1cm}

\\
MSC2010 database: 47A40, 81U05, 47B36




\end{abstract}

\vspace{.5cm}


\section{Introduction}

\label{sec-intro}

We  study scattering theory for the one-dimensional matrix  Schr\"odinger  Hamiltonians  on $\mathbb{Z}$ of the form 
\begin{align}\label{introeq}
H = H_{0}  + V, 
\end{align}
where $H_{0}$ is {(an energy shift of)} the discrete Laplace operator and $V$ is a self-adjoint  matrix multiplication operator {with finite first moment. See Section~\ref{sec-Hams} below for a detailed definition of $H$.} This model is widely  used in the context of low-energy phenomena in solid state physics. Moreover, $H$ is a tridiagonal operator, also called a Jacobi operator, which is the discrete analogue of a Sturm-Liouville operator.  Its analysis is connected to orthogonal (matrix) polynomials and, via its spectral theory,  to {matrix-valued measures} on the real line.  {Many authors have studied direct and inverse scattering theory for this class of operators, see for example Case and Kac \cite{Cas}, Serebryakov \cite{Ser}, Aptekarev and Nikishin \cite{APT,Nik}, Geronimo \cite{Ger} and Guseinov \cite{Gus,Gus3}. More recent contributions are \cite{Ego,Vol}.}

\vspace{.1cm}

The energy is parametrized by $E = z + 1/z $, for {$z \in \overline{\mathbb{D}}\setminus\{0\}$ where $\DM=\{ z \in \mathbb{C} \: : \:  |z| < 1 \} $} and this parametrization is used in  the  analysis of the scattering matrix which is constructed from the Jost solutions, namely asymptotically free generalized eigenfunctions. This permits to extend  the transmission coefficients {meromorphically} to $  \mathbb{D}\setminus \{ 0\}  $ using the properties of the Wronskian. We compute the limit {of the transmission coefficients} when $z$ tends to $1$  and give an explicit  expression for them (with a limit taken in 
$\overline{\mathbb{D}}$). We also prove that the limits of the reflection coefficients exist when $z$ tends to $1$ 
(these coefficients are only defined for $z \in \mathbb{S}^1 $, and $\mathbb{S}^1$ is regarded as a subset of $\mathbb{C}$).  The limits when $z$ tends to $-1$ can be studied in a similar fashion and, for this reason, we omit them. The above results are useful for inverse scattering theory and the proof of Levinson's theorem, and we will address these problems in future works.

\vspace{.1cm}

{Scattering} theory for matrix Schr\"odinger operators on $\ZM$ of the form \eqref{introeq}, but with compactly supported $V$, is studied in \cite{BCS}. However, in \cite{BCS} a different analytical approach is followed, namely the solutions of the eigenvalue equation of $H$ are calculated in terms of the transfer matrix, whereas here the Volterra equation is used. This implies a major difference and there is practically  no intersection between  the proofs of  the present   manuscript and the proofs in \cite{BCS}.  Moreover, a compactly supported potential admits meromorphic continuations of the scattering matrix to the whole complex plane, something that is not possible in the present framework.  Nevertheless, we essentially stick with the notations of \cite{BCS} in this paper, but there are some new technical objects that are only addressed in the present analysis. 

\vspace{.1cm}

Let us briefly comment on other earlier contributions. 
{In the continuous scalar case,  \cite{Cha} contains everything from Jost solutions to low-energy behavior of reflection and transmission coefficients (see Ch. XVII).  In  Sect. XVII4.4 and in  page 381, some ideas and references for the matrix-valued case are presented. Much of this is now driven by interest in completely integrable systems, in particular, by studies of the matrix-valued KdV and Toda equations. The inverse scattering theory approach was developed in great detail in \cite{Cal,Cal2,Mar,Olm,WDT,G}, and  the half-line case was addressed in \cite{New}. For further references, see \cite{Dei,Fir,Ges,Bon}.}

\vspace{.1cm}

{The low-energy behavior of the scattering matrix is studied in \cite{Kla} for the scalar continuous case and in \cite{HKS} for the scalar discrete case.} For continuous matrix-valued Schr\"odinger operators,  {  there are works by  Wadati and Kamijo \cite{WDT}, Mart\'inez Alonso and Olmedilla \cite{Mar,Olm}, Corona-Corona \cite{G}, Newton and Jost  \cite {New} as well as } by Aktosun, Klaus, Van Der Mee and Weder \cite{AKV,AKW,AW}. We would like to stress here that our proof  for the limit of the scattering matrix at the band edges is shorter and closely tied to a conceptual treatment of the Wronskian. Actually we do not use the Jordan decomposition as in \cite{AKV}. {There are also works in the non-linear context, see for example \cite{Cal,Cal2,CN2018,N2018}}. Bound states for the half-space version of the  discrete matrix Schr\"odinger equation with compactly supported potentials have been analyzed in \cite{ACP}.  


%
\subsection{Mathematical framework and main results}
\label{sec-Hams}

In this manuscript, $\CM^L$ denotes the $L$-dimensional complex vector space and $\mathcal{M}_{d\times e}(\mathbb{C}) =  \mathcal{M}_{d\times e}$ the vector space of matrices  with $d$ rows and $e$ columns and with coefficients in $\mathbb{C}$. Furthermore, $M^*$ denotes the adjoint of a matrix $M$, that is, its conjugate transpose.

\subsubsection{Hamiltonians}

Scattering theory compares asymptotic time evolution of two systems. The simpler one is called free system and the other is the interaction system. In this paper, the free system is described by the   discrete  Laplacian on the Hilbert space $  \ell^2(\mathbb{Z}, \mathbb{C}^L) $.
It is given by
\begin{align}\label{eq-H0}
(H_0 \phi)(n) \; := \;   \phi(n+1) \; + \; \phi(n-1), \hspace{2cm} \phi \in  \ell^2(\mathbb{Z}, \mathbb{C}^L).
\;
\end{align}
We denote by $V  \in   \mathcal{M}_{L\times L}(\mathbb{C})^{\mathbb{Z}}   $  the interaction, which is a matrix-valued multiplication operator  
 defined by
\begin{align}\label{VVV}
(V \Psi)(n) :=   V(n)\Psi(n), \hspace{2cm} \Psi \in  \ell^2(\mathbb{Z}, \mathbb{C}^L),
\end{align}
and assume that $V(n) = V(n)^*$ for every $n \in \mathbb{Z}$, and that       
\begin{align}\label{shortrange}
\sum_{n \in \mathbb{Z}} \| n  V(n)\| < \infty. 
\end{align}
Now the interaction Hamiltonian is defined by \eqref{introeq}, namely $H=H_0 + V$, with domain  $ \ell^2(\mathbb{Z}, \mathbb{C}^L) $. 
We denote  by $\mathcal{F} :    \ell^2(\mathbb{Z}, \mathbb{C}^L)  \to  L^2([- \pi, \pi], \mathbb{C}^{L}) $ the Fourier transform and by $\mathcal{F}^{-1}$ its inverse. They are given by
\begin{align} \label{Fourier}
\mathcal{F}(\phi)(k) =  \frac{1}{\sqrt{2 \pi}} \sum_{n } e^{ik n } \phi(n),  \hspace{2cm}   \mathcal{F}^{-1}(\psi)(n)
= \int_{- \pi }^{\pi}   \frac{1}{\sqrt{2 \pi}}  e^{-ik n } \psi(k).
\end{align}
A direct calculation leads us to
\begin{align}\label{H0fourier}
\mathcal{F}H_0  \mathcal{F}^{-1} \psi (k) = (e^{ik} + e^{-i k})\psi(k) = 2 \cos(k) \psi(k).
\end{align}
As the Fourier transform is unitary, then {the spectrum of $H_0$ is $ \sigma(H_0) = [-2,2] $ and} it is purely absolutely continuous. The essential spectrum of $H$ is {thus} $ [-2, 2] $ (see Section XIII.4 in \cite{RaS}).  

\vspace{.1cm}

This paper studies stationary scattering theory. Then, naturally,  eigenvalue equations for $H$ and $H_0$ are relevant in this manuscript. As usual, we study generalized eigenvalues and, moreover, do not only address real energies, but also study generalized eigenvectors corresponding to complex energies. We use the same symbols $H$ and $H_0$ to denote the operators defined (with the same expressions as above) either on $(\mathcal{M}_{L \times L}  )^{\mathbb{Z}}$    or  $(\mathbb{C}^L)^{\mathbb{Z}}$ and parametrize the eigenvalues in the form 
\begin{equation}
\label{eq-Ez}
E = z + 1/z
\;,
\qquad
z \in \mathbb{C} \setminus \{ 0\}.
\end{equation}
Then, the eigenvalue equations take the form 
\begin{align}\label{solutions}
Hu = Eu,    \\
H_{0}u = Eu,   \label{Fsolutions} 
\end{align}
and generally we take   $u \in (\mathcal{M}_{L \times L})^{\mathbb{Z}}$.  In the context of this article we simply call solutions the functions $u$ satisfying \eqref{solutions}.  The solutions of \eqref{Fsolutions} are referred {to as the} free solutions. Of particular importance are the Jost solutions, which are solutions with prescribed data at $- \infty $ or $\infty$  given by free solutions. They are essential objects for the scattering matrix. We also study solutions with prescribed data {at $0$ and $1$}, because they are important in technical elements in our proofs.

\subsubsection{Jost solutions}

\vspace{.2cm}

Jost solutions are the key ingredient for the construction of the scattering matrix. They are solutions of the system that behave as free waves (plane waves) at infinity. {Standard properties and the construction of these solutions can be found in the books \cite{AGR,Cha}.} 
In \eqref{H0fourier},  $k$ might be seen as a momentum and the Fourier modes, {namely} the functions
$ n \mapsto  e^{i n k }$ represent plane waves. These are generalized eigenvalues of $H_0$ and satisfy the equation
\begin{align}\label{planewavesenergy}
H_0 e^{ i n  k} \alpha =   (e^{ik} + e^{-i k}) e^{ in  k}\alpha = 2\cos(k)e^{in  k}  \alpha ,
\end{align}
for every $   \alpha \in \mathbb{C}^L$. 
In the previous equation we identify, as usual, the function $ n \mapsto  e^{i n k }  \alpha $ with $ e^{i n k }   \alpha $. For $k \notin \{-\pi, 0 ,\pi  \}  $,  the functions $\{e^{in k} \alpha, e^{-ik n}  \alpha :   \alpha \in \mathbb{C}^{L}\}$ define a  $ 2 L $-dimensional vector space and, as $H_0$ has only two discrete derivatives, they generate all solutions of \eqref{planewavesenergy} (see also \eqref{Fsolutions}). Of course, $2\cos(k) = E$ is interpreted as the energy of the corresponding plane waves.  In this generalized sense, the  solutions of the time dependent Schr\"odinger equation
$$ i \frac{d}{dt }\phi = H_0 \phi   $$
are of the form
$$  e^{-i (E t - kn)}  \alpha , \hspace{2cm} e^{-i (E t + kn)}  \alpha.    $$
For positive energies $E$, the first wave function above moves in the direction of $k$ and the second in the direction of $-k$. For negative energies it is the other way around.  Heuristically, we understand a wave traveling to the right
allowing $n$ to be real (as in the continuous case) and looking at the equation $ Et - kn = 0  $ (taking the phase to be zero): for positive $E$ and $k$ (for example), a positive increment in time leads to a positive increment in position.

\vspace{.1cm}

It is convenient to change our notation and take $ z = e^{ik} $. With this notation,  we take $z \in \mathbb{C} $ and not only $ z = e^{ik} $ as before, {\it i.e.} we analytically continue the solutions. From the discussion above, we obtain that a complex number  $ z = e^{ik} $ represents  a wave traveling to the right
if  its real and imaginary parts have the same sign. Otherwise, it travels to the left.  This implies that if $ z $ represents a wave traveling to the right, then $1 /z$ represents a wave traveling to the left, and vice versa.
Notice that this holds true only when  $ z = e^{ik} $, if this is not fulfilled then there is no interpretation for the direction of traveling.
We analyze matrix valued solutions in order to consider all vector valued solutions at once. We set (for $z \ne 0$) 
\begin{align}\label{matrixsolution}
   u_0^z(n)
   \;: =\;
   z^n\,\one ,
\end{align}
where  $\one$ is the identity in $\mathcal{M}_{L\times L}$. It satisfies the complex extension of \eqref{planewavesenergy}, notably with $E=z+1/z$,
\begin{align}\label{planewavesmatrix}
H_0 u_0^z =  E u_0^z,  
\end{align}
where we use matrix multiplication,
and the vector valued solutions of \eqref{planewavesenergy} are of the form $  u_0^z {a}  $.
Notice that for any fixed value of the energy $E$ and $ z \notin \{-1,0, 1 \} $, the columns of $ u_0^z $ and
$  u_0^{1/z} $ generate all solutions. This is not the case for $z = 1$ or $z =  -1 $ (which correspond to $E= 2$, or $E = -2$, respectively) because in this situation $ u_0^z = u_0^{1/z}  $. In order to provide all solutions, also for $E= 2$ and $E= -2$, we define
\begin{align}\label{vuno}
v_0^{\pm} (n) := (\pm 1)^n\,n.
\end{align}
Then, the columns of the matrices $  u_0^{\pm 1} $ and $ v_0^{\pm}   $ generate the space of free solutions (generalized eigenvectors) of \eqref{planewavesenergy}  for $E =   \pm 2$. As mentioned above, the generalized eigenvectors of  $ H_0  $ are free waves (or plane waves). Jost solutions are generalized eigenvectors of $H$ that behave as plane waves away from the interaction. {They are introduced in the next definition}. Their existence is proved  in Section~\ref{sec-Jost}.

\begin{defini}[Jost Solutions]\label{Jost}
For every $z  \in \overline{\mathbb{D}} \setminus \{  0  \} $,  we denote by $ u_{+}^{z}  $
, $u_{-}^{1/z}$ the $\mathcal{M}_{L\times L}$-valued solutions of
\begin{align} \label{hamiltoninan}
H u^{z}_{+} = E u^{z}_{+},   \hspace{1cm} H u^{1/z}_{-} = E u^{1/z}_{-}{,}    \hspace{3cm} E = z + 1/z{,}
\end{align}
satisfying, as {$n \to + \infty$ and $n \to - \infty$ respectively},
\begin{align}\label{wasu}
u_+^z(n) = z^n (  \one + o(1)) , \hspace{1cm}
u_-^{1/z}(n) =  z^{- n} (  \one + o(1)).
\end{align}
Moreover, for $z = 1$, we denote by $ v_{\pm}^{z}  $ the $\mathcal{M}_{L\times L}$-valued solutions of 
\begin{align}
H v^{z}_{\pm} = E v^{z}_{\pm},   \hspace{3cm}  E=2,
\end{align}
satisfying  as $n \to \pm \infty$
\begin{align}\label{wasv}
v^{z}_{\pm }(n) =  n  (  \one + o(1)).
\end{align}
\end{defini}

\subsubsection{The scattering matrix}

Due to the asymptotic behavior of {the Jost} solutions (see Equation \eqref{wasu}), the columns of the matrix $( u^{z}_{\pm}, u^{1/z}_{\pm} )$ are linearly independent for $ z \in  \mathbb{S}^1 \setminus  \{-1, 1 \} $ and, therefore,  they form a basis of solutions.  The same holds true for   $( u^{ 1}_{\pm}, v^{1}_{\pm} )$. This implies that there are matrices
$M^z_\pm,  N^z_\pm  \in\Mm_{L\times L} $ such that
\begin{align}\label{MmasmenNmasmen}
u_+^z=u_-^zM^z_++u_-^{\overz}N^z_+, \hspace{1cm}
u_-^{\overz}=u_+^zN^z_-+u_+^{\overz}M^z_-.
\end{align}
Moreover, it will be proved that the matrices $M^z_\pm$ have a meromorphic  continuation to $\mathbb{D}$ (see Equation \eqref{eq-MNId0}). 
Assuming  that $ M^z_\pm $ are invertible, we can rewrite these equations as
\begin{align}\label{tresfromH}
u_+^z T^z_+ = u_-^z - u_-^{\overz}R^z_{+} , \hspace{1cm}
u_-^{\overz}T^z_- =    u_+^{\overz}  - u_+^zR^z_-,
\end{align}
where
\begin{align}\label{transref}
T^z_\pm=(M^z_\pm)^{-1},  \hspace{2cm}R^z_\pm=- N^z_\pm(M^z_\pm)^{-1}{,}
\end{align}
are the transmission and reflection coefficients, respectively.
The interpretation of \eqref{tresfromH}  in the case that $ z = e^{ik}\notin
\{  -1, 1 \} $, {corresponding to a wave traveling to the right,}   is the following (we only describe the first equation in \eqref{tresfromH}):  the  incoming wave $u_-^{z}$ produces the
outgoing wave $ u_+^z T^z_+  $ traveling to the right ({\it i.e.}, a transmitted wave) and the  outgoing wave $  u_-^{\overz}R^z_{+}   $  traveling to the left ({\it i.e.}, a reflected wave).  The relation  between transmitted and  reflected waves is described by the scattering matrix.

\begin{defini}
\label{def-Scat}
For any  $ z \in  \mathbb{S}^1 \setminus  \{-1, 1 \} $,  the scattering matrix $\Ss^{z} \in \mathcal{M}_{2L \times 2L}$  is defined by
$$
\Ss^z
\;=\;
\begin{pmatrix}
T_+^z & R^z_- \\ R^z_+ & T^z_-
\end{pmatrix}  
\;=\;
\begin{pmatrix}
(M^z_+)^{-1} & - N^z_-(M^z_-)^{-1} \\
-N^z_+(M^z_+)^{-1} & (M^z_-)^{-1}
\end{pmatrix} .
$$
\end{defini}

Notice that matrices $M^z_{\pm}$ are {indeed} invertible, see Proposition \ref{lem-Minvertible}.
In the case that $z = e^{ik}$ represents a wave traveling to the right, then
$  u^z_- $  and $ u^{\overz}_+ $ are incoming and $ u^{z}_+$  and $  u^{\overz}_-  $ are
outgoing. In this case,
the scattering matrix expresses the incoming Jost solutions $u^z_-$ and $u^{\overz}_+$ in terms of the outgoing ones $u^{z}_+$ and $ u^{\overz}_-$:
\begin{equation}
\label{eq-SMatDef}
\begin{pmatrix}
u^z_- & u^{\overz}_+
\end{pmatrix}
\;=\;
\begin{pmatrix}
u^{z}_+ & u^{\overz}_-
\end{pmatrix}
\Ss^z
\;.
\end{equation}

\subsection{Main results}

The following theorem is proved in Theorem~\ref{T} below.

\begin{theo}
\label{prop-ScatProp} 
There is a neighborhood of $ 1 $ such that, for every $z$ in this neighborhood with $z \in \overline{\mathbb{D}}$,   the matrices $M^z_{\pm}$ are invertible.  Moreover, the limits
\begin{align}
{T_{\pm}^1:=}\lim_{z \to 1}  T_{\pm}^z 
\end{align}
exist, where the limits are taken in $ \overline{\mathbb{D}}$, and they have explicit expressions (see \eqref{900}). The kernels and images of $T_{\pm}^1$ can be explicitly calculated (see \eqref{901}) and are tightly connected to half-bound states.
The limits
\begin{align}
{R_{\pm}^1:=}\lim_{z \to 1}  R_{\pm}^z 
\end{align}
of the reflection coefficients  $R_{\pm}^z$ exist, where the limits are taken in $ \mathbb{S}^1$, and they have explicit expressions (see \eqref{903}). The kernels and images of $ \one - R_{\pm}^1 $ can be explicitly calculated (see \eqref{904}). 
\end{theo}
\begin{rem}
Theorem \ref{prop-ScatProp} is also valid if one takes the limits $z \to -1$, and the proofs are the same. In order to simplify notations, we focus on the case
$z \to 1$.    
\end{rem}

\section{Solutions  }

\subsection{{Free solutions  with prescribed data on \texorpdfstring{$0$}{0} and \texorpdfstring{$1$}{1}}} 

\begin{defini}\label{FreeSolo1}
For every  $z\in \overline{ \mathbb{D}} \setminus\{0\} $, we denote by $s^z$ and $\tau^z$ the {scalar solutions $s^z,\tau^z\in\CM^\ZM$} of \eqref{Fsolutions} such that  $s^z(0) = 0$, $s^z(1) = 1$ and $\tau^z (0) = 1 = \tau^{z}(1)$. 
\end{defini}

Explicitly, one can verify that 
\begin{align}\label{eq-sz}
s^z(n)= \left\{ \begin{array}{cc}
	\frac{1}{z - z^{-1}}(z^{n}-z^{-n})  =  \frac{z}{z+1} \sum_{j= -n}^{n-1}  z^j, &   z^2\neq 1 \\
	(\pm1)^{n+1}n,  & z=\pm1
	\end{array} \right.
\end{align}	
and (for $z \ne -1$)
\begin{align}\label{tau}
\tau^z(n) = \frac{z^{n}+z^{-n +1 }}{z+1}. 
\end{align}	
Let us introduce the notation $  D(w; r):=
\{ z \in \mathbb{C} \: : \:  |z- w| < r  \} $.

\begin{lemma}\label{Remarkk}
	For all $z \in  \overline{\mathbb{D}} \cap D(1; 1/2) $, the following holds true:
	\begin{enumerate}
\item[(i)] \begin{align}\label{PPPH}
 |s^z(n) | \leq C |n||z|^{-|n|} , \hspace{.2cm} \frac{|s^z(n) - s^1(n)|}{|z-1|} \leq C n^2  |z|^{-|n|},    \hspace{.2cm}  s^z(n) - s^1(n) = O( |z-1|^2  ) , 
\end{align} 
as $z$ tends to 1, where the $O $ symbol does depend on $n$, but $C$ does not. 
\item[(ii)]
\begin{align} \label{PPPH0}
\Big |\frac{   s^z(n-j) -  s^1(n-j) - (s^z(n) - s^1(n)  )}{z-1}\Big | \leq C |nj ||z|^{-|n|}.      
\end{align}
\item[(iii)]
\begin{align}\label{esttau}
|\tau^z(n) - \tau^1(n)| = O(|z-1|^2), 
\end{align}
where the $O$ symbol depends on $n$.  
\item[(iv)]
\begin{align}\label{esttau1}
|  \tau^{z}(n) |\leq  C  |z|^{-|n|},     \hspace{1cm} \Big |  \frac{  \tau^{z}(n) - \tau^1(n) }{ z-1 }  \Big | \leq  C |n | |z|^{-|n|} .
\end{align}

\end{enumerate}

\end{lemma}

\begin{proof}
Since  
\begin{align}\label{PPPH1}
 s^z(n) =   \frac{z}{z+1} \sum_{j= -n}^{n-1}  z^j, 
\end{align} the left {bound} in \eqref{PPPH} is obvious. {Equation}  \eqref{PPPH1} implies that   
$  s^z(n) $  is analytic and its derivative is uniformly bounded
by a constant times $n^2  |z|^{-|n|}$. Then, an application of  the mean value theorem  yields  the middle {bound} in \eqref{PPPH}. A direct  calculation shows that the  derivative of $  s^z(n)  $ vanishes at $ z=1$, and Taylor's theorem implies the right {bound} in \eqref{PPPH}.

\vspace{.1cm}

The derivative of $    s^z(n-j) -  s^z(n) = \frac{z}{z+1} \Big (  \sum_{j= -(n- j)}^{n-j- 1} z^j  
 -     \sum_{j= -n}^{n - 1} z^j \Big ) $  is uniformly bounded by a constant times $|z^{-n}| n j$ and, therefore, the mean value theorem shows \eqref{PPPH0}.

\vspace{.1cm}

An elementary calculation gives that $\frac{d}{dz} \tau^z (n)|_{z=1} = 0$, and this together with Taylor's theorem implies \eqref{esttau}. 

\vspace{.1cm}
  
The left identity of \eqref{esttau1}  is obvious from the definition
of $\tau^z$. The right identity of \eqref{esttau1} follows again {from the mean} value theorem, since  
the derivative of $\tau^z$ is uniformly bounded by  $  C |n | |z|^{-|n|} $. 
\end{proof}

\subsection{Jost Solutions}
\label{sec-Jost}

\begin{lemma}[Jost Solutions]\label{soljost}
The Jost solutions   $u^z_{+} $,  $u^{1/z}_{-} $ as defined in Definition \ref{Jost} exist, for every $z \in $  $ \overline{\mathbb{D}}\setminus\{0\}$.  
	Moreover, for every  $n$, the functions  $u_{+}^{z}(n),u_{-}^{1/z}(n)$ are  holomorphic on   $\mathbb{D}\setminus\{0\}$ and continuous on $\overline{\mathbb{D}}\setminus\{0\}$. The following Volterra equations are satisfied: 
	\begin{equation}\label{6}
	\begin{aligned}
	u_{+}^z(n)&=z^{n}\mathbf{1} - \sum_{j=n+1}^{\infty} s^z(j-n)V(j)u_{+}^z(j) , \qquad n\in \mathbb{Z},\\
	u_{-}^{1/z}(n)&=z^{-n}\mathbf{1} + \sum_{j=-\infty}^{n-1} s^{1/z}(j-n)V(j)u_{-}^{1/z}(j) , \qquad n\in \mathbb{Z},		 
	\end{aligned}	
	\end{equation}
where $s^z$ is defined in Definition \ref{FreeSolo1}. 
\end{lemma} 

\begin{proof} 
The result follows from the Theorem \ref{volterra}: Taking $g=\mathbf{1}, K^z(n,j)=-z^{j-n}s^z(j-n)V (j)$ and $M(j)=j\|V(j)\|$, we obtain  a solution $\tilde{u}_+^z$ to the equation 
(for $n \in \mathbb{N}$)
 \begin{align}\label{PH1}	\tilde{u}_{+}^z(n)=\mathbf{1} - \sum_{j=n+1}^{\infty} s^z(j-n)V(j)z^{j-n}\tilde{u}_{+}^z(j).
\end{align} 
  A direct computation using \eqref{PH1} shows that   $u_+^z(n)=z^n\tilde{u}_+^z(n)$ solves the {Schr\"odinger} equation 
  $   u_+^z(n-1)  +V(n)u^z_+(n) + u^z_+(n+1) =(z+1/z)u^z_+(n) $, 
  for $n \geq 2$  (this is  proved in Lemma \ref{PPH}).  
   For $n\in \mathbb{Z}^-\cup\{0\}$, we  recursively fit     Equation \eqref{solutions} defining : 
   $u_+^z(n-1)=(z+1/z)u^z_+(n)-V(n)u^z_+(n)-u^z_+(n+1)$.
   The construction of the other solution is similar.
\end{proof}

\begin{lemma}
	The solutions $v_\pm^1$  introduced in Definition \ref{Jost} exist.
\end{lemma}
\begin{proof}
 Theorem \ref{volterra} implies that there is a solution $\tilde{v}^1_+  $  to the  Volterra equation (for $n \geq  N$)
\begin{align}\label{PH2}	
	 \tilde{v}^1_+(n)=\mathbf{1}+\frac{1}{n}\sum_{j=N}^{n}j^2V(j)\tilde{v}^1_+(n) +\sum_{j=n+1}^{\infty}jV(j)\tilde{v}^1_+(n),
\end{align}	 
	  where $N\in \mathbb{N}$ is such that $\sum_{j=N}^{\infty}j\|V(j)\|<1/2$. Here we set $g=\mathbf{1},$ $$K(n,j)=  \left\{ \begin{array}{cc}
	 	jV(j), &   j\geq n+1 \\
	 	\frac{j^2}{n}V(j), & N\leq j\leq n
	 	\end{array} \right.$$ and $M(j)=j\|V(j)\|$.  
A direct calculation  using \eqref{PH2}  shows that  $   v^1_+(n)= n \tilde{v}^1_+(n)$ solves the Schr\"odinger   equation  $   v_+^1(n-1)  +V(n)v^1_+(n) + v^1_+(n+1) =2v^1_+(n) $, for $n \geq N+1$ (this is carried out using similar methods as in the proof of Lemma \ref{PPH}). 
For $n \leq N $, we  recursively fit     Equation \eqref{solutions} defining  
  $v_+^1(n-1)=2 v^1_+(n)-V(n)v^1_+(n)-v^1_+(n+1)$.
 The construction of the other solution is similar. 
\end{proof}

\vspace{.1cm}

A key point of this paper is the study of Jost solutions when the spectral parameter $z$ tends to $1$. It turns out to be more accessible to control the behavior of solutions  with prescribed  data
on $0$ and $1$  as $z\to 1$ in a detailed manner. Via Wronskian identities this ultimately allows to deal with Jost solutions and the behavior of the scattering matrix as $z\to 1$.

\subsection{Solutions with prescribed data {at \texorpdfstring{$0$}{0} and \texorpdfstring{$1$}{1}}}

\begin{lemma}\label{lem01}
Let $a$, $b$  $\in \mathcal{M}_{L \times L}$. 
	For every  $z\in \overline{ \mathbb{D}} \setminus\{0\} $, the solution  $\Psi^z$ of \eqref{solutions} such that  $\Psi^z(0) = a$, $\Psi^z(1) = b$ satisfies the following equations: for every $n \in \mathbb{N}$,  
	\begin{equation}\label{phi_1,natp}
	\begin{aligned}
	\Psi^z(n)&=	 s^z(n)(b-a)+ \tau^z(n) a-\sum_{j= 1}^{n-1}   s^z(n-j)V(j)\Psi^z(j), 		
	\end{aligned}
	\end{equation}
	and for every  $n \in \mathbb{Z}^-\cup\{0\}$
	\begin{equation}\label{phi_1_entp}
	\begin{aligned}
	\Psi^z(n)&= s^z(n)
	  (b-a)+ \tau^z(n) a +\sum_{j= n+1}^{0} 
	  s^z(n-j)
	  V(j)\Psi^z(j). 		
	\end{aligned}
	\end{equation}	
Moreover, for every fixed $n$, $\Psi^{z}(n)$ is  holomorphic on   $\mathbb{D}\setminus\{0\}$ and continuous on $\overline{\mathbb{D}}\setminus\{0\}$.

\end{lemma}
\begin{proof}
The result follows from Lemma \ref{var_par}, taking $A=-(z+1/z)\mathbf{1}, B(n)=-V(n), S_1=s^z$ and $S_2=\tau^z$. The analyticity and continuity follows from the analyticity and continuity of $s^z(n)$ and $\tau^z(n)$.
\end{proof}

\begin{lemma}\label{boundm}
Let $\Psi^z$ be as in Lemma \ref{lem01}.
For all $z \in  \overline{\mathbb{D}} \cap D(1; 1/2) $, the estimates
\begin{align}
\|\Psi^z(n)  \| \,\leq\, & C |n| |z|^{-|n|}, \label{uno} \\ 
\label{dos}
\|\Psi^z(n) - \Psi^{1}(n)  \| \,= \, & O( |z-1|^2 ),    \hspace{1cm}  z\to 1, 
\end{align}
hold, where the $O$ symbol depends on $ a$, $b$ and $n$, but not on $z$ and  $C$ depends on $ a$ and $b$, but not on $z$ and $n$.  
\end{lemma}

\begin{proof}
Let $ n \in \mathbb{N}$, it follows  from \eqref{phi_1,natp} and Lemma \ref{Remarkk} that  
\begin{align}\label{x0}
\|	\Psi^z(n)\| |z|^n {\frac{1}{n}} &=|z|^n {\frac{1}{n}}  \Big(	| s^z(n)(b-a)+ \tau^z(n) a-\sum_{j= 1}^{n-1}   s^z(n-j)V(j)\Psi^z(j)| \Big )   \\ \notag  & \leq
C \Big(   1 +  \sum_{j= 1}^{n-1}  \|  j V(j)\| \  \| \Psi^z(j)\| \ |z|^{j} \frac{1}{j} \Big),
\end{align}
  and, therefore, Gronwall's Lemma (see  Lemma \ref{gronwall}) {combined with \eqref{shortrange} } yields \eqref{uno}. For negative $n$, the argument is similar, using \eqref{phi_1_entp}.   
Now we prove \eqref{dos} for $n \in \mathbb{N}$, the case of $n$ negative is proved similarly. The proof uses  induction on $n$. By definition, $\Psi^z(1)-\Psi^1(1)=0$.  Suppose that for every $j<n,$ $\Psi^z(j)-\Psi^1(j)=O(|z-1|^2)$. It follows from Lemma \ref{lem01} that
\begin{align}\label{x1}
 \| \Psi^z(n) - \Psi^1(n)\|\, \leq \,&  	\| s^z(n) - s^1(n) \| \ \|b-a\|+ \|\tau^z(n) - \tau^1(n)\| \ \|a\| 	
 \\  \notag & + \sum_{j= 1}^{n-1}  \| s^z(n-j)  V(j) \| \ \|\Psi^z(j) -\Psi^1(j) \|
 \\ & \notag +  \sum_{j= 1}^{n-1}  \| s^z(n-j) - s^1(n-j)\| \ \| V(j) \| \ \|\Psi^1(j) \|. 
\end{align}
{Equation}  \eqref{x1},  Lemma \ref{Remarkk} and \eqref{uno} { for $z=1$} imply that there is a function $\iota(n, z)\geq 0$ such that $  \iota(n, z)  = O(|z-1|^2) $ and 
\begin{align}\label{x3}
 \| \Psi^z(n) - \Psi^1(n)\| \leq & \iota(n, z)  + C n |z|^{-n}\sum_{j= 1}^{n-1}  
 \|V(j)\| \ \|\Psi^z(j) -\Psi^1(j) \|. 
\end{align}
Since each element in the sum {on the right is} $O(|z-1|^2)$, it follows that  $\| \Psi^z(n) - \Psi^1(n)\| =O(|z-1|^2)$, and the argument is completed by induction.
\end{proof}

\vspace{.1cm}

Of particular importance are here solutions with the initial condition $a=u^1_+(0), b=u^1_+(1).$

\begin{defini}
\label{def-Phi}
For all $z \in  \overline{\mathbb{D}} \cap D(1; 1/2) $, 
the solution described in Lemma \ref{lem01} with $ a =u_+^1(0)$ and $ b=u_+^1(1)$ is denoted by $\Phi^z\in{(\mathcal{M}_{L\times L})}^{\mathbb{Z}}$. 
\end{defini}
	
Notice that $\Phi^1=u_+^1$.  In this case,
we have that (see Lemma \ref{soljost})	
\begin{equation}
\label{eq-b-a}
b- a = \sum_{j=1}^{\infty}V(j)u_+^1(j)=  \sum_{j=1}^{\infty}V(j)\Phi^1(j).
\end{equation} 
Let us set
	$$ d(n) : = \sum_{j=n}^{\infty}V(j)\Phi^1(j)  .$$
It follows from Lemma	\ref{lem01} that  
for every $n \in \mathbb{N}$,  
	\begin{align}\label{phi_1,nat}	
	\Phi^z(n)= \,&	 s^z(n) d(n) + \tau^z(n) a-\sum_{j= 1}^{n-1}   s^z(n-j)V(j)(\Phi^z(j)- \Phi^1(j) ) 
	\\ \notag & - \sum_{j= 1}^{n-1}  ( s^z(n-j) -  s^z(n) )V(j)  \Phi^1(j)  , 
	\end{align}
and for every  $n \in \mathbb{Z}^-\cup\{0\}$
	\begin{equation}\label{phi_1_ent}
	\begin{aligned}
	\Phi^z(n)&= s^z(n)(b-a) + \tau^z(n) a +\sum_{j= n+1}^{0} 
	  s^z(n-j)
	  V(j)\Phi^z(j),		
	\end{aligned}
	\end{equation}

{The following result shows that $\Phi^z$ satisfies essentially the same bounds as $s^z$ and $\tau^z$ given in Lemma~\ref{Remarkk}.}

\begin{lemma}[Regularity 1]\label{princest1}
There is a constant  $C \in \mathbb{R}$ independent of $n$ and $z$ such that,   for all $z \in  \overline{\mathbb{D}} \cap D(1; 1/2) $,   
	\begin{equation}\label{phi_1-z}
 \left\|\frac{\Phi^z(n)-\Phi^1(n)}{z-1}\right\|\leq Cn|z|^{-n} , \qquad 	\forall n \in \mathbb{N} . 
	\end{equation}
\end{lemma} 

\begin{proof}
It follows from \eqref{phi_1,nat}, that 
\begin{align}\label{luc}
	\Phi^z(n)  -\Phi^1(n) \,=\, &	 (s^z(n)- s^1(n)) d(n) + (\tau^z(n) -  \tau^1(n)  ) a-\sum_{j= 1}^{n-1}   s^z(n-j)V(j)(\Phi^z(j)- \Phi^1(j) ) 
	\notag \\  & - \sum_{j= 1}^{n-1}  \Big ( s^z(n-j) -  s^1(n-j)  - \big ( s^z(n) 
	-  s^1(n) \big ) \Big )V(j)  \Phi^1(j)  . 
\end{align}
{Equation} \eqref{luc}, Lemma \ref{Remarkk}, \eqref{shortrange},  and the fact that $nd(n)$ is uniformly bounded (see \eqref{shortrange})  implies that
\begin{align}\label{fin}
\frac{|z|^n}{n}\Big \| \frac{	\Phi^z(n)  -\Phi^1(n)}{ z-1 } \Big \| \leq 
C   \Big [  1 +   \sum_{j= 1}^{n-1} j \|V(j)\|   \ \frac{|z|^j}{j} \ \Big \| \frac{	\Phi^z(j)  -\Phi^1(j)}{ z-1 } \Big \|   \Big ],
\end{align}
here we use that $ \Phi^1(j) = u^1_{+}(j)  $ is uniformly bounded for $j \geq 0$ (which is a consequence of the definition of the Jost {solution in question}). 
{Equation}  \eqref{fin} and Gronwall's Lemma (see Lemma \ref{gronwall}) {combined with \eqref{shortrange} } imply \eqref{phi_1-z}. 
\end{proof}

\vspace{.1cm}

Lemma \ref{boundm} (with $\Phi^{1}  $ playing the role of {$\Psi^1$})   implies that the series
\begin{align}\label{nimodo}
\sum_{ j =- \infty}^\infty V(j)\Phi^{1}(j) = 
\sum_{ j =- \infty}^\infty V(j)u_+^{1}(j),
\end{align}
converges {to a matrix in $\Mm_{L\times L}$}.  This  matrix   plays an important role in {the} proofs {as it is connected to the Wronskian of $u_+^1$ and $v_-^1$,} see Lemma \ref{lemWigual}.  In Lemma  \ref{NM} (see Definition \ref{kernels}), we prove that for every vector $\xi$ that belongs to its kernel, the sequence $ {(   u_{+}^1(j) \xi )_{j \in \mathbb{Z}}}$ is bounded. The corresponding  states 
  $  (  u_{+}^1(j) \xi )_{ j \in \mathbb{Z} }$  are  called half-bound states.

\begin{lemma}[Regularity 2]\label{princest2}
Let $\xi$ belong to the kernel of  \eqref{nimodo}.  
 There is a constant  $C \in \mathbb{R}$, independent of $n$ and $z$, such that,  for  all $z \in  \overline{\mathbb{D}} \cap D(1, 1/2) $,
	\begin{equation}\label{phi_1-z22}
	  \left\|\frac{1}{z-1}(\Phi^z(n)-\Phi^1(n)) \xi \right\|\leq C|n||z|^{n} ,
\qquad \forall \,n \in \mathbb{Z}^- \cup \{  0\}.
	\end{equation}
\end{lemma} 

\begin{proof}
In this case, (see {Equation~\ref{eq-b-a}} and recall that $\xi$ belongs to the kernel of \eqref{nimodo})	
\begin{align} \label{ba}
(b- a)\xi =   \sum_{j=1}^{\infty}V(j)\Phi^1(j)\xi =  -    \sum_{j=- \infty}^{ 0}V(j)\Phi^1(j)\xi.
\end{align}
Hence let us set
$$ d_{-}(n) : = - \sum_{j=-\infty }^{n }V(j)\Phi^1(j)  .$$
Due to \eqref{shortrange} one has
\begin{equation}
\label{eq-d-bound}
|nd_-(n)|\leq \sum_{j=-\infty }^{n }|j|\,\|V(j)\|\,\|\Phi^1(j)\|<\infty.
\end{equation}
It follows from Lemma	\ref{lem01} {applied to $\Phi^z$} and \eqref{ba} that for every  $n \in \mathbb{Z}^-\cup\{0\}$ 
	\begin{align}\label{phi_1_entptt}
	\Phi^z(n) \xi = & s^z(n)d_{-}(n) \xi + \tau^z(n) a \xi +\sum_{j= n+1}^{0} 
	  s^z(n-j)
	  V(j) (\Phi^z(j)  - \Phi^1(j)  )  \xi
	  \\ & \notag  + \sum_{j= n+1}^{0} 
	  \big ( s^z(n-j) -  s^z(n) \big )
	  V(j)  \Phi^1(j)    \xi  .		
	\end{align}
Then, we have that 
	\begin{align}\label{phi_1_entpppt}
\Big \| \frac{|z|^{-n}}{n}	 & \frac{1}{z - 1 }(\Phi^z(n) \xi -\Phi^1(n) \xi) \Big \|  \leq   \frac{|z|^{-n}}{n}	  \frac{1}{z - 1 } \Big [   \|(s^z(n)  - s^1(n)  )d_{-}(n) \xi \| \notag 
\\ \notag & +\| (\tau^z(n)  -  \tau^1(n)  ) a \xi \| +\sum_{j= n+1}^{0} 
	  \| s^z(n-j)
	  V(j) (\Phi^z(j)  - \Phi^1(j)  )  \xi \|
	    \\  &  + \sum_{j= n+1}^{0} 
	 \| \big ( s^z(n-j) - s^1(n-j)  - ( s^z(n) - s^1(n) ) \big )
	  V(j)  \Phi^1(j)    \xi\|  \Big ] .		
	\end{align}
Bounding the first summand on the right is bounded using the second estimate from \eqref{PPPH} and \eqref{eq-d-bound}, the second summand with \eqref{esttau} and the fourth summand with \eqref{PPPH0}, one deduces
\begin{align}\label{finnn}
\frac{|z|^{-n}}{n}\Big \| \frac{	\Phi^z(n)  -\Phi^1(n)}{ z-1 } \xi \Big \| \leq 
C   \Big [  1 +   \sum_{j= n+1}^{0} j \|V(j)\| \   \frac{|z|^{-j}}{j} \ \Big \| \frac{	\Phi^z(j)  -\Phi^1(j)}{ z-1 } \xi\Big \|   \Big ],
\end{align}
because $ \Phi^1(j) \xi   $ is uniformly bounded (see Lemma \ref{NM} and Definition \ref{kernels}). 
{Equation}  \eqref{finnn} and Gronwall's  (see Lemma \ref{gronwall}) imply 
\eqref{phi_1-z22}. 
\end{proof}

\section{The scattering matrix}
\label{sec-Scat}

\begin{defini}[Wronskian]
 For two functions $u,v:\ZM\to \mathcal{M}_{L\times L}$ and $n\in\ZM$, the Wronskian is defined by
\begin{equation}
\label{eq-WronskiDef}
W(u,v)(n)
\;=\;
{i}\,\big(u(n+1)^*v(n)\,-\,u(n)^*v(n+1)\big)
\;\in\mathcal{M}_{L\times L}
\;.
\end{equation}
\end{defini}

It is easy to see that $W(u,v)^*=W(v,u)$ and that  for matrix solutions  satisfying $Hu=Eu$ and $Hv=\overline{E}v$, the Wronskian $W(u,v)$ is independent of $n$. In these cases we omit the argument $n$. 
For $ z \in \overline{\mathbb{D}} $,
the Wronskians of the Jost solutions can be evaluated  using  that they are independent of   $n$: taking the limits, either  $n  \to \infty$ or 
$n \to - \infty $, we find that
\begin{align}\label{Wid1}
W(u^{\zb}_+,u^z_+) = 0 = W(u^{1/ \zb}_-,u^{1/z}_-) , 
\end{align}
and if, additionally,  $z \in \mathbb{S}^1 \setminus \{ -1, 1 \}$
\begin{align}\label{Wid2}
W(u^z_\pm,u^z_\pm)
\;=\;
(\nu^z)^{-1}\,\one
\;,
\end{align}
where
\begin{equation}
\label{eq-NuDef}
\nu^z\;=\;\frac{{i}}{z-\overzz}
\;.
\end{equation}

Next recall that for  $ z \in  \mathbb{S}^1 \setminus  \{-1, 1 \} $, it is possible to decompose the  states $u^z_+$ and $u^{\overz}_-$ on the basis $(u_-^z,u_-^{\overz})$  and $(u_+^z,u_+^{\overz})$, respectively, and that the  matrices $M^z_\pm$ and $N^z_\pm$ are defined by 
\begin{equation}
\label{eq-udecomp}
u_+^z
\;=\;u_-^zM^z_+\;+\;u_-^{\overz}N^z_+
\;,
\qquad
u_-^{\overz}
\;=\;u_+^zN^z_-\;+\;u_+^{\overz}M^z_-
\;.
\end{equation}
{Equation~\ref{Wid2}}
leads to
\begin{align}
\label{eq-MNId0}
\begin{split}
&
M^z_+
\;=\;
\nu^z\,
W(u_-^{\overzb},u^z_+)
\;,
\\
&
N^z_+
\;=\;
-\,\nu^z\,
W(u_-^{1/z},u^z_+)
\;,
\\
&
N^z_-
\;=\;
\nu^z\,
W(u_+^{z},u^{\overz}_-)
\;,
\\
& M^z_-
\;=\;
-\,\nu^z\,
W(u_+^{\zb},u^{\overz}_-)
\;.
\end{split}
\end{align}
This shows that $M^z_\pm$  can be extended to analytic functions  on $\mathbb{D}\setminus\{0\}$ which are continuous on  $\overline{\mathbb{D}}\setminus\{-1,0,1\}$.
{Equations} \eqref{eq-MNId0} imply that
\begin{align} 
&
(N^z_+)^*
\;=-N^z_-\;,
\qquad
(M^z_+)^*
\;=\;
M^{\zb}_-
\;,
\label{eq-MNId2}
\end{align}
where the first equation holds true for  $ z \in  \mathbb{S}^1 \setminus  \{-1, 1 \} $  and the second can be extended to    $\overline{\mathbb{D}}\setminus\{-1,0,1\}$. 

\begin{lemma}
For every $ z \in  \mathbb{S}^1 \setminus  \{-1, 1 \} $, the following identities hold true: 
	\begin{align}
	\label{I2}
	(M^{z}_-)^*M^z_-\;&=\;\one\,+\,(N^{z}_-)^*N^z_-\;,\\\label{I3}
	M^{z}_+N^z_-\;& =\;-\,N^{\overz}_+M^z_-\;,\\\label{I6}
		(M^{z}_+)^*M^{z}_+\;&=\;\one\,+\,(N^{z}_+)^*N^z_+,\\ 
		 \label{I5}
		 	M^{z}_-N^{z}_+\;&=\;-\,N^{\overz}_-M^z_+\;. 
	\end{align}
\end{lemma}
\begin{proof}
We notice that for every matrix $  M \in \mathcal{M}_{L\times L}$ and every solutions $u, v$, 
\begin{align}\label{to1}
W(uM,(v+w))=& M^*(W(u,v)+W(u,w)),  \\  \notag  W(u+v,wM)= & (W(u,w)+W(v,w))M. 
\end{align}
 First we prove \eqref{I6}. It follows from Equations \eqref{Wid2} and \eqref{eq-udecomp} that
	\begin{align}\label{to2}(\nu^z)^{-1}\mathbf{1}=W(u^z_+,u^z_+)=W(u_-^zM^z_+\;+\;u_-^{\overz}N^z_+,u_-^zM^z_+\;+\;u_-^{\overz}N^z_+). 
	\end{align}
Expanding the right hand side of \eqref{to2} and using Equations \eqref{Wid1}, \eqref{Wid2} and \eqref{to1}, we get
$$(\nu^z)^{-1}\mathbf{1}=(\nu^z) ^{-1}(M^z_+)^*M^z_+-(\nu^z)^{-1}(N^z_+)^*N^z_+,$$ where $\nu^{1/z}=-\nu^z$ was used. This implies \eqref{I6}. Equation \eqref{I2} is obtained in similar manner by expanding $W(u_-^z,u_-^z)$. Now let us prove $\eqref{I5}$.  It follows from Equations \eqref{Wid1} and \eqref{eq-udecomp} that
\begin{align}\label{to3}
0=W(u^{1/z}_+,u^z_+)=W(u_-^{1/z}M^{1/z}_+\;+\;u_-^{z}N^{1/z}_+,u_-^zM^z_+\;+\;u_-^{\overz}N^z_+).
\end{align}
Expanding the right hand side of \eqref{to3} and using Equations \eqref{Wid1}, \eqref{Wid2} and \eqref{to1}, we get
  $$0=-(\nu^z)^{-1}(M_+^{1/z})^*N_+^z+(\nu^z)^{-1}(N_+^{1/z})^*M_+^z=-(\nu^z)^{-1}M_-^zN_+^z-(\nu^z)^{-1}N_-^{1/z}M_+^z,$$ where the last equality follows from \eqref{eq-MNId2}. Equation \eqref{I3} is obtained in similar manner expanding $W(u_-^{1/z},u_-^z)$. 
\end{proof}

\begin{proposi}
\label{lem-Minvertible}
For $z\in\SM^1\setminus\{-1,1\}$, $M^z_\pm$ is invertible and $\Ss^z$ is unitary  

\end{proposi}

\begin{proof}
The invertibility of $M^z_\pm$ follows from Equations \eqref{I2} and \eqref{I6}. Now we prove the unitarity. The off diagonal terms of $  (\Ss^z)^*\,\Ss^z  $ are 
	(see Definition \ref{def-Scat})
	\begin{align} 
	-((M_+^z)^{-1})^*N_{-}^z (M_-^z)^{-1} 
	-((M_+^z)^{-1})^*(N_{+}^{z})^* (M_-^z)^{-1} , \\
	-((M_-^{z})^{-1})^*(N_{-}^{z})^* (M_+^z)^{-1} - 
	((M_-^{z})^{-1})^* N_{+}^{z} (M_+^z)^{-1}
	\end{align}
	and they vanish by \eqref{eq-MNId2}. The  diagonal terms are 
	\begin{align}
	((M_+^z)^{-1})^* (1 + (N_+^z)^* N_+^z )  (M_+^z)^{-1} ,  \\ \notag
	(  (M_-^z)^{-1})^* (1 + (N_-^z)^* N_-^z )  (M_-^z)^{-1} ,
	\end{align}
	and they are  both equal to  $\mathbf{1}$, see \eqref{I6} and \eqref{I2}.
	This proves the unitary of $\Ss^z$.
\end{proof}

\section{Analysis of the Wronskian }

It follows from \eqref{eq-MNId0} and 
Definition \ref{def-Scat} that the Wronskian is tightly connected to the scattering matrix. In particular, the invertibility of $M^z_{\pm}$ is essential for its definition.   Since the purpose of the present paper is the analysis of the scattering matrix as $z$ tends to $1$, it is crucial to study the behavior of $  W(u_{-}^{1/\overline{z}}, u_{+}^z) =
 (\nu^z)^{-1} M^{z}_{+}  $ as $z$ tends to $1$ (the study of $M_-^z$ is carried out using \eqref{eq-MNId2}). Regularity properties of $u_{+}^z$  as $z$ tends to $1$ are {thus} relevant. As stated above, this will be deduced from the regularity results on $\Phi^z$. Indeed, it turns out that these properties of $ \Phi^z $ allow  to identify lower order terms of  $  W(u_{-}^{1/\overline{z}}, u_{+}^z)$ with respect to $|z-1|  $. This {holds  because} 
$  W(u_{-}^{1/\overline{z}}, u_{+}^z)$ can be written in terms of 
  $  W(\Phi^{\overline{z}}, u_{+}^z)$  and $  W(u_{-}^{1/\overline{z}}, \Phi^z)$ (see Lemma \ref{lema_w}). Then, most of this section is devoted to the study of   $  W(\Phi^{\overline{z}}, u_{+}^z)$  and  $  W(u_{-}^{1/\overline{z}}, \Phi^z)$.  

\vspace{.1cm}

From the  definition of $u_+^1$, we know that $u_{+}^1(j)$ tends to {${\bf 1}$}  
as $j$ tends to infinity. Then, for large enough $j$,    $u_{+}^1(j)$ is invertible.  In order to simplify notations, we assume that   $u_{+}^1(1)$ 
is already invertible. This does not imply any restriction because, translating the origin, we can always take it for granted.  

\begin{as}
We assume, without loss of generality, that  $u_{+}^1(1)$ 
is invertible. 
\end{as}

\begin{lemma}\label{lema_w}
	For  every $z \in \overline{\mathbb{D}}\setminus\{0\}$, it follows that 
	\begin{equation}\label{idwron}
  W(u_{-}^{1/\overline{z}}, u_{+}^z) =
  u_-^{1/\bar{z}}(1)^*({u_+^1(1)}^*)^{-1} W(\Phi^{\overline{z}}, u_{+}^z)
  +  W(u_{-}^{1/\overline{z}}, \Phi^z)   u_+^1(1)^{-1}u_+^z(1).
	\end{equation}
\end{lemma}
\begin{proof}
The result follows from an expansion of the right hand side of \eqref{idwron}  using Definition~\ref{def-Phi} of $\Phi^z$ and the definition of the Wronskians, evaluated on $n = 0$,  and the identity (see  \eqref{Wid1}) 
	\begin{equation*}
	{(u_+^1(1)^*)}^{-1}u_+^1(0)^* = u_+^1(0)u_+^1(1)^{-1}. 
	\end{equation*}
Notice that, by definition, $\Phi^z(j) =u_+^1(j)  $, for $j \in \{0, 1 \}$.
\end{proof}

\begin{proposi}\label{PW1}
The following formula holds true
\begin{align}	
 W(\Phi^{\overline{z}}, u_{+}^z) =i (1-z)\mathbf{1}+o(|z-1|),
\end{align}
as $z$ tends to $1$ in $\overline{\mathbb{D}}$. 
	
\end{proposi}

\begin{proof} 
{Equation}  \eqref{phi_1,natp} {for $\Psi^z=\Phi^z$ together with the bounds from Lemma~\ref{Remarkk} } imply that  {for $n\in\NM$}
\begin{align}
{|z|^n}	\|\Phi^z(n)\| \leq C_{z}+ \sum_{j= 1}^{n-1} C_z   {|z|^j} \ \|V(j)\| \  \|\Phi^z(j)\|, 
\end{align}
for some positive function $C_{z}$ {depending on} $z$ (that blows up as $z$ tends to $1$
because in Equation~\eqref{eq-sz} for $s^z$ a cancellation of the factor $\frac{1}{1-z}$ implies growth in $n$). Then, 
 Gronwall's lemma implies that $ {|z|^n}	|\Phi^z(n)|  $ is bounded (with respect to $n$, for positive $n$).
 This implies that (see {Equation}  \eqref{wasu})
 \begin{align*}
 W(\Phi^{\overline{z}}, u_{+}^z) & = \lim_{n \to \infty } 
 i( \Phi^{\overline{z}}(n+1)^* z^n -  \Phi^{\overline{z}}(n)^* z^{n+1})
\\ 
&=  i(b-a)^*- i (z-1)a^* - i \sum_{j=1}^{ \infty} z^{j}\Phi^{\overline{z}}(j)^*V(j),  
\end{align*}  
where { in the second step we used \eqref{phi_1,natp} and  Definition~\ref{FreeSolo1} of $s^z$ and $\tau^z$}.  On the other hand, 
\begin{align}\label{chon1}
 \sum_{j=1}^{ \infty} z^{j}(\Phi^{\overline{z}} - \Phi^{1}  )(j)^*V(j) = o(|1-z|). 
\end{align}
Indeed, Equations \eqref{phi_1-z} and 
\eqref{shortrange} imply that the  series multiplied by $\frac{1}{1-z}$ is bounded by  a summable function that does not depend on $z$,  and \eqref{dos} implies that each term of the sum  multiplied by $\frac{1}{1-z}$ tends to zero. Hence interpreting the series as an integral with respect to a counting measure, Lebesgue's dominated convergence theorem shows \eqref{chon1}. Furthermore, since $ \Phi^{1} $ is bounded for $j > 0$ (because $  \Phi^{1}(j) = u^1_+(j) $, for $j >1$) it follows that 
\begin{align} \label{chon2}
\sum_{j=1}^{ \infty} (z^{j} -  1 - j (z - 1)  ) \Phi^{1}  (j)^*V(j) = o(|1-z|).
\end{align}
Then, we obtain that 
 \begin{align*}
 W(\Phi^{\overline{z}}, u_{+}^z) =  i(b-a)^*-i (z-1)a^* -  i  \sum_{j=1}^{ \infty}  (1 + j(z-1))  \Phi^{1}(j)^*V(j) + o(|1-z|) . 
\end{align*} 
This last equation and the fact that (see  \eqref{6})
	\begin{equation}\label{l67}
	\begin{aligned}
	(b-a) &= \sum_{j=1}^{\infty}V(j)u_+^1(j)=  \sum_{j=1}^{\infty}V(j)\Phi^1(j),\\
	a&= \mathbf{1} - \sum_{j=1}^{\infty}jV(j)u_+^1(j)=\mathbf{1} - \sum_{j=1}^{\infty}jV(j)\Phi^1(j),
	\end{aligned}
	\end{equation}
imply the desired result.	\end{proof}

\begin{lemma}\label{lemWigual}
The following formula holds true
\begin{align}
W(u_-^1,u_+^1) = - i  \sum_{j=-\infty}^{\infty}V(j)u_{+}^1(j)  .
\end{align}

\end{lemma}
\begin{proof}
It follows from  Lemma \ref{soljost} that
\begin{equation}\label{2266}
\begin{aligned}
&u_{+}^1(n)=\mathbf{1} - \sum_{j=n+1}^{\infty} (j-n)V(j)u_{+}^1(j), \qquad n \in \mathbb{Z},\\
&u_{-}^1(n)=\mathbf{1} + \sum_{j=-\infty}^{n-1} (j-n)V(j)u_{-}^1(j), \qquad  n \in \mathbb{Z}.
\end{aligned}
\end{equation} 
Since $ |j V(j)|  $ tends to zero as $j$ tends to minus infinity (see \eqref{shortrange}) and there is a constant $ C  $ such that
$ |u^1_{+}(j)| \leq C |j|  $, for $j \leq 0$, (see Lemma \ref{boundm}), it follows that for $n\to\infty$ 
	\begin{equation*}
	\begin{aligned}
	W(u_-^1,u_+^1)(n) 
	= & i \left(\mathbf{1}+\sum_{j=-\infty}^{n} (j-n-1)u_{-}^1(j)^*V(j)\right)\left(\mathbf{1} - \sum_{j=n+1}^{\infty} (j-n)V(j)u_{+}^1(j)\right)\\
	&- i \left(\mathbf{1}+\sum_{j=-\infty}^{n-1} (j-n)u_{-}^1(j)^*V(j)\right)\left(\mathbf{1}-\sum_{j=n+2}^{\infty} (j-n -1)V(j)u_{+}^1(j)\right)	\\
	=  & -i\sum_{j=n+2}^{\infty}V(j)u_{+}^1(j) + i\sum_{j=-\infty}^{n}u_-^1(j)^* V(j)\sum_{j=n+1}^{\infty} (j-n)V(j)u_{+}^1(j) +o(1),
	\end{aligned}
	\end{equation*}
Here, we used that $u^1_{\pm}( \pm j)$ is bounded for $j \in \mathbb{N}$ and that   $\lim_{ m \to \infty}
\sum_{j = -\infty}^{-m}  |j| \|
V(j)\|  +  \sum^{\infty}_{j = m}  |j| \|
V(j)\| = 0  $. Taking the limit $n \to - \infty$ yields the desired result.    	 
\end{proof}

We define some technical objects that will be used in Section \ref{HB}. 

\begin{defini}\label{kernels}
We introduce the notations
\begin{align}\label{PH3}
\Nn:={\Ker} (  W(u_-^1,u_+^1) ),  \hspace{1cm}  \Ll:={\Ker} (W(u_+^1,u_-^1)). 
\end{align}
The generic case is referred as  $\mathcal{N}=\{0\}$, otherwise one speaks of the exceptional case.
In Lemma \ref{NM}, 
a characterization of  $\mathcal{N}$ and  $ \Ll $ in terms of half-bound states is presented, in particular we prove that  
$  (u_{+}^{1}(j) \xi)_{j \in \mathbb{Z}} $ is bounded, for every 
$  \xi \in \mathcal{N} $ ({\it i.e.}, it is a half-bound state, see Section \ref{HB}). For $\xi \in \Nn$, let us define
\begin{align}\label{gammasn}
\Gamma \xi  :=  \Big (\xi -  \sum_{j    =  - \infty}^{\infty} j V(j)  u_{+}^{1}(j) \xi \Big ).
\end{align}
In Lemma \ref{defGamma} another characterization of $\Gamma$ {is given and} in Lemma~\ref{biye} it is shown 
that $\Gamma$ is a bijection from $\mathcal{N}$ onto $\Ll$.  We denote by
\begin{align}\label{pene}
P_{\mathcal{N}},   \quad P_{\Ll}, \quad P_{\mathcal{N}^{\perp}},\quad P_{\Ll^{\perp}}
\end{align}
 the orthogonal projections onto $\mathcal{N}$, $\Ll$, $ \mathcal{N}^{\perp} $ and $\Ll^\perp$ respectively. 
\end{defini}

\begin{proposi}\label{la 4.3}
The following formula holds true
\begin{align}\label{teto1}
  W(u_{-}^{1/\overline{z}}, \Phi^z) =  W(u_{-}^{1}, u_{+}^1)  +o(1),
\end{align}
as $z$ tends to $1$ in $\overline{\mathbb{D}}$. 
Moreover, if $\xi \in \mathcal{N}   $, then 
\begin{align}\label{teto2}
  W(u_{-}^{1/\overline{z}}, \Phi^z)\xi  =    i(1-z) \Gamma\xi  +o(|z-1|),
\end{align}
as $z$ tends to $1$ in $\overline{\mathbb{D}}$.
\end{proposi}

\begin{proof}
{Equation}  \eqref{teto1} follows from the continuity of the functions 
$z \mapsto u_{-}^{1/\overline{z}}$ and $ z \mapsto \Phi^z $, and the fact that 
$  \Phi^1 =  u^1_+  $. Now we  prove \eqref{teto2}. 
 {Equation}  \eqref{phi_1_entp} implies that
\begin{align}
|z^{-n}|	|\Phi^z(n)| \leq C_{z}+ \sum_{j= n+ 1}^{0} C_z   |z^{-j}| \ \|V(j)\| \  \|\Phi^z(j)\|, 
\end{align}
for some positive function $C_{z}$ of $z$ (that blows up as $z$ tends to $1$). Then, 
 Gronwall's lemma implies that $ |z^{-n}|	|\Phi^z(n)|  $ is bounded (with respect to $n$, for  $n \leq 0$).
 This implies that (see {Equation}  \eqref{wasu}) 
 \begin{align*}
 W(u_{-}^{1/\overline{z}}, \Phi^z)  
 &= \lim_{n \to -  \infty } 
 i(  z^{-n-1}\Phi^{z}(n)  - z^{-n} \Phi^{z}(n+1) )
\\
& = - i(b-a)+i \frac{1-z}{z}a -  i \sum_{j=- \infty}^{ 0} z^{-j}V(j)\Phi^{z}(j),  
\end{align*}  
where \eqref{phi_1_entp} was used.  Utilizing   \eqref{dos}, \eqref{phi_1-z22}    and 
\eqref{shortrange} we deduce that (here we use again Lebesgue's dominated convergence theorem)
\begin{align}\label{chon1t}
 \sum_{j=- \infty}^{ 0} z^{-j}V(j)
 (\Phi^{z}(j)-  \Phi^1(j) ) \xi  = o(|1-z|). 
\end{align}
Since $ \Phi^{1}(j)\xi  $ is bounded (see Lemma \ref{NM}), it follows that for $j \leq  0$ 
\begin{align} \label{chon2t}
\sum_{j=-\infty}^{0} (z^{-j} -  1 - j ( 1 - z )  )  V(j) \Phi^{1}  (j)\xi = o(|1-z|).
\end{align}
Then 
 \begin{align*}
  W(u_{-}^{1/\overline{z}}, \Phi^z)\xi = - i(b-a) \xi +i \frac{1-z}{z} a \xi-   i  \sum_{j=-\infty }^{ 0}  (1 + j(1-z))  V(j)  \Phi^{1}(j) \xi  +  o(|1-z|).  
\end{align*} 
The  desired result from this last equation and \eqref{l67}. Notice that we use that $  \frac{1- z}{z} - (1-z) =  o(|1-z|) $ and that by assumption (see Lemma  \ref{lemWigual} and Definition \ref{kernels})  $ \sum_{j = - \infty}^\infty V(j)\Phi^1(j) \xi =0 $ and recall that $u_{+}^1 = \Phi^1$ . 
\end{proof}

\vspace{.1cm}

The  operator {$  u_-^{1/\overline{z}}(1)^*({u_+^1(1)}^*)^{-1} $} that appears in \eqref{idwron} plays an important role because, when $z = 1$, it operates on half-bound states (see Remark \ref{obs_inv}) and it is present in the scattering matrix, in the limit when $z$ tends to $1$. For this reason, we also introduce a notation for this object.    

\begin{defini}
We denote   
	\begin{equation}\label{omega}
	\Omega := u_+^1(1)^{-1}u_-^1(1)\in\Mm_{L\times L}.
	\end{equation}
\end{defini}

A different characterization of $\Omega$ is given in Remark \ref{obs_inv}.
 
\begin{proposi}\label{Tmain}
There exist  functions {$X,Y:\DM\to\Mm_{L\times L}$} such that
\begin{align}\label{main}
  W(u_{-}^{1/\overline{z}}, u_{+}^z)  = 
 \Big (& i (1 - z) ( \Omega^* + \Gamma ) P_{\mathcal{N}}
  +   W(u_{-}^{1}, u_{+}^1) P_{\mathcal{N}^{\perp}}   
    + X(z) + Y(z)P_{\mathcal{N}^{\perp}}  \Big ) \notag  \\ &\cdot   u_+^1(1)^{-1}u_+^z(1),
\end{align}
and $X(z) = o(|z-1|)$ and $ Y(z) = o(1)$, as $z $ tends to $1$ in 
$\overline{\mathbb{D}}$. 
\end{proposi}

\begin{proof}
The result follows  from the continuity of $  u_-^{1/\bar{z}}(1)  $ and $  u_+^z(1) $,  Propositions \ref{PW1} and \ref{la 4.3} and Lemma \ref{lema_w}.  
\end{proof}

\section{Proof of the main result}

\subsection{Half-bound states}\label{HB}

In this section we study half-bound states at the threshold energy corresponding
to the spectral parameter $z= 1$. They are solutions of \eqref{solutions} with {$E=2$} that are bounded. These solutions play a fundamental role in the limit {of the scattering matrix} as 
$z$ tends to $1$.

\begin{lemma}\label{asinsol1}
	Let $u\in {(\mathbb{C}^L)}^{\mathbb{Z}}$ be a solution of \eqref{solutions}, for $z =1$. 
	The following items are equivalent: 
	\begin{enumerate}
		\item[(i)]  $u$ is  $o(n)$ for $n \to  +\infty$.
		\item[(ii)] $u$  is bounded as $n $ tends to $ +\infty$.
		\item[(iii)] $u$ converges as $n $ tends to $ +\infty$.
	\end{enumerate}
	Moreover,  $u$ is $o(1)$ for $n \to +\infty$, if and only if $u=0$.
\end{lemma}
\begin{proof}
It follows from Equations \eqref{wasu} and \eqref{wasv} that the columns  of 
$ u_{+}^1 $ and $v_{+}^1$  form a basis of all solutions. Then, there are $\alpha,\beta \in \mathbb{C}^L$ such that  $$u=v_+^1 \alpha + u_+^1 \beta.$$ 
The asymptotic behavior of $ u_{+}^1  $ and $v_{+}^1 $ at $\infty$ yields the desired result. 
\end{proof}
\begin{rem}\label{asinsol2}
The previous lemma remains valid if we replace    $+\infty$ by $-\infty$.
\end{rem}
\begin{lemma}\label{asinu}
	The next equations hold true:
	$$	u_{+}^1(n) =  n(i W(u_-^1,u_+^1)  + o(1)) ,\qquad n\to -\infty.$$
	$$	u_{-}^1(n) =  n(-i W(u_+^{1},u_-^1) + o(1)) ,\qquad n\to +\infty.$$
\end{lemma}

\begin{proof}
It was proved in Lemma \ref{lemWigual} that
\begin{align}
W(u_-^1,u_+^1) =  -i\sum_{j=-\infty}^{\infty}V(j)u_{+}^1(j),
\end{align}
and, similarly we deduce that 
\begin{align}
W(u_+^{1},u_-^1) = i\sum_{j=-\infty}^{\infty}V(j)u_{-}^1(j).
\end{align}
{Equation}  \eqref{6} implies that 
	\begin{equation*}
	\begin{aligned}
	u_+^1(n+1)-u_+^1(n) &= -\sum_{j=n+2}^{\infty}(j-n-1)V(j)u_+^1(j) + \sum_{j=n+1}^{\infty}(j-n)V(j)u_+^1(j)\\
	&= \sum_{j=n+1}^{\infty}V(j)u_+^1(j)\\
	& =i  W(u_-^1,u_+^1) - \sum_{j=-\infty}^{n} V(j)u_{+}^1(j).
	\end{aligned}
	\end{equation*}
	Thus,  $u_+^1(n+1)-u_+^1(n) \to   i   W(u_-^1,u_+^1) $, as $n \to -\infty$ (recall that $|u_+^1(n)|\leq C|n| $  {due to Equation}  \eqref{uno}).
	This implies that
	$$\frac{1}{n}(u_+^1(n+1)-u_+^1(1))=\frac{1}{n}\sum_{j=1}^{n}u_+^1(j+1)-u_+^1(j) \to  i   W(u_-^1,u_+^1) $$
	and, therefore, 
	$$\frac{u_+^1(n)}{n}\; \to \;i   W(u_-^1,u_+^1) , \qquad n \to \infty.$$ 
This proves the first equality. The proof of the second is similar.  
\end{proof}

The next result establishes a connection between the subspace $\Nn$ and $\Ll$ introduced in Definition \ref{kernels} and the half-bound states.

\begin{lemma}[Half-Bound States]\label{NM}
The next equations hold true:
\begin{align}
\mathcal{N}= \{\xi \in \mathbb{C}^L : u_+^1\xi \ \text{is} \ \text{bounded} \} ,  \hspace{1cm}  \Ll =
\{\chi \in \mathbb{C}^L : u_-^1\chi \ \text{is} \ \text{bounded}\}. 
\end{align}
Moreover, since  $W(u_-^1,u_+^1)^* = W(u_+^1,u_-^1) $ by definition, it follows that 
\begin{align}
\mathcal{N} = (W(u_+^1,u_-^1)\mathbb{C}^L)^{\perp},  \hspace{1cm} \Ll = (W(u_-^1,u_+^1)\mathbb{C}^L)^{\perp},
\end{align}
and, therefore,  $\dim(\mathcal{N})=\dim(\Ll)$. 
\end{lemma}

\begin{proof}
	Take $\xi\in \mathcal{N}$,  then Lemmas \ref{asinsol1} and  \ref{asinu} 
	yield that $ u_{+}^1 \xi$ is bounded.     
	Now taking $\xi\in \mathbb{C}^L$ such that $u_+^1\xi$ is bounded,  Lemma \ref{asinu} implies that 
	$$0=\lim_{n\to-\infty} \frac{u_+^1(n)\xi}{n}= i  W(u_-^1,u_+^1)  \xi .$$ Therefore the first equality follows. The proof of the second equality is similar. 
\end{proof}

\begin{lemma}\label{defGamma}
	For every  $\xi \in \mathcal{N}$,  
	\begin{equation}\label{54}
	\Gamma\xi = \lim_{n \to -\infty} u_+^1(n)\xi.
	\end{equation}  
\end{lemma}

\begin{proof}
Let $\xi \in \mathcal{N} = {\Ker(W(u_{-}^1, u_{+}^1))}$, we calculate using 
	Lemma \ref{lemWigual} and {Equation}  \eqref{6}:  
	\begin{equation*}
	\begin{aligned}
\xi - u^1_{+}(n)\xi &= 	
	\sum_{j=n+1}^{\infty}(j-n)V(j)u_+^1(j)\xi
	\\ &=\sum_{j=n+1}^{\infty}jV(j)u_+^1(j)\xi-n\sum_{j=n+1}^{\infty}V(j)u_+^1(j)\xi\\
	&=\sum_{j=n+1}^{\infty}jV(j)u_+^1(j)\xi-n\left(  iW(u_{-}^1, u_{+}^1) -\sum_{j=-\infty}^{n}V(j)u_+^1(j)\right)\xi\\
	&=\sum_{j=n+1}^{\infty}jV(j)u_+^1(j)\xi+n\sum_{j=-\infty}^{n}V(j)u_+^1(j)\xi.
	\end{aligned}
	\end{equation*}
Notice that $ u_{+}^1(j) \xi$ is bounded because $\xi \in \mathcal{N}$ (see Lemma \ref{NM}). 
From  the last equation and definition of $\Gamma$ (see Definition \ref{kernels}),  we have that (see also \eqref{shortrange})
	\begin{equation}
	\begin{aligned}
	\lim_{n\to-\infty}u_+^1(n)\xi&=\xi-\lim_{n\to-\infty}\sum_{j=n+1}^{\infty}(j-n)V(j)u_+^1(j)\xi\\
	&=\xi-\sum_{j=-\infty}^{\infty}jV(j)u_+^1(j)\xi=\Gamma\xi,
	\end{aligned}
	\end{equation}
completing the proof.
\end{proof}

\begin{lemma}\label{biye}
	$\Gamma$ is a linear isomorphism from  $\mathcal{N}$ to $\Ll$.
\end{lemma}
\begin{proof}
Taking  $\xi\in \mathcal{N}$ and $\chi = \Gamma\xi$,
	it follows from Lemma \ref{defGamma} and Equation \eqref{wasu} that 
	\begin{equation*}
	\lim_{n \to -\infty} u_+^1(n)\xi-u_-^1(n)\chi = 0.
	\end{equation*}
	Then, Remark  \ref{asinsol2} implies that
	\begin{equation}\label{24}
	u_+^1(n)\xi=u_-^1(n)\chi, \ n \in \mathbb{Z}.
	\end{equation}
	We deduce that  $u_-^1\chi$ is bounded and, therefore,   $\chi\in \Ll$ and it follows that $\Gamma \mathcal{N} \subset \Ll$.  Let  $\chi \in \Ll$ and  $\xi = \lim_{n \to \infty} u_-^1(n)\chi$. As above, we obtain that $u_+^1(n)\xi=u_-^1(n)\chi$ for $ n\in \mathbb{Z}$ and, therefore,  $\Gamma\xi = \chi$. This proves the subjectivity and as $\Ll$ and $\mathcal{N}$ have the same dimension, that $\Gamma $ is bijective.      
\end{proof}

\begin{rem}\label{obs_inv}
	It follows from  \eqref{24} that 
	\begin{equation}\label{gama}
	\Gamma = (u_-^1(n)^{-1}u_+^1(n)){\big|_\mathcal N },
	\end{equation}
whenever $u_-^1(n)^{-1}$ exists.  
As explained above {it can be assumed} without loss of generality that this is the case when $n=1$.  Let us also recall that  
	\begin{equation}\label{omega_1}
	\Omega = u_+^1(1)^{-1}u_-^1(1),
	\end{equation}
	and notice that 
	\begin{equation}  \label{testa}
	\Omega{\big|_{\Ll}}= \Gamma^{-1}.
	\end{equation}
\end{rem}

\subsection{{Band edge} limit of the scattering matrix}


{Let  $\{e_{1}, \cdots, e_{L}\} $ be} an orthonormal   basis of $\mathbb{C}^L$ such that the first $d$ vectors form a basis of $  \Ll  $ and the {last  $L-d$ vectors form} a basis of 
 $\Ll^{\perp} = W(u_{-}^{1}, u_{+}^1 ) \mathbb{C}^L $, see  Lemma \ref{NM}. We take {another orthonormal  basis  $\{v_1, \cdots, v_L  \}$ of $\mathbb{C}^L$}  such that 
 the first $  d $ vectors form a  basis of $\mathcal{N}$ and the last $ L-d  $ vectors form a basis of $\mathcal{N}^{\perp}$.
Then define
\begin{align}\label{change}
P : =  \begin{pmatrix}
e_1 & e_2 & \cdots & e_L
\end{pmatrix}^*, \hspace{1cm} Q:=  \begin{pmatrix}
v_1 & v_2 & \cdots & v_L
\end{pmatrix}.
\end{align}   
We recall that $ P_{\Ll}$ and $P_{\Ll^{\perp}}$ are the projections onto $\Ll$ and $\Ll^{\perp}$, respectively.   Then 
\begin{align}\label{chich}
 P_{\Ll} \Big ( i (1 - z) ( \Omega^* + \Gamma ) P_{\mathcal{N}}
  +   W(u_{-}^{1}, u_{+}^1) P_{\mathcal{N}^{\perp}}   
     \Big )P_{\mathcal{N}} =  i (1 - z) P_{\Ll}  (  \Omega^* + \Gamma ) P_{\mathcal{N}} 
\end{align} 
and   $     P_{\Ll}  (  \Omega^* + \Gamma ) P_{\mathcal{N}}   $ defines a  bijection between $\mathcal{N}$ and $\Ll$ (see Lemma \ref{NM}): in view of Lemma \ref{NM}   it is enough to  prove that it is   injective. This holds true because  $\Gamma: \mathcal{N} \to \Ll $ is a bijection and, for every 
   $\xi \in \mathcal{N}$,  (see {Equation}  \eqref{testa}) 
\begin{align*}
\langle \Gamma \xi,  P_{\Ll}(\Omega^* + \Gamma)\xi   \rangle
& =\langle P_{\Ll}\Gamma \xi,  (\Omega^* + \Gamma)\xi   \rangle=\langle \Gamma \xi,  (\Omega^* + \Gamma)\xi   \rangle 
\\
& {=}  \;\langle  \Omega \Gamma \xi,  \xi   \rangle +  \| \Gamma \xi \|^2 =    \|  \xi \|^2 +  \| \Gamma \xi \|^2.
\end{align*}
 Moreover, 
\begin{align}\label{chich1}
 P_{\Ll^{\perp}} \Big ( i (1 - z) ( \Omega^* + \Gamma ) P_{\mathcal{N}}
  +   W(u_{-}^{1}, u_{+}^1) P_{\mathcal{N}^{\perp}}   
     \Big )P_{\mathcal{N}^{\perp}} = 
     P_{\Ll^{\perp}}    W(u_{-}^{1}, u_{+}^1) P_{\mathcal{N}^{\perp}}   
\end{align} 
defines a  bijection between $\mathcal{N}^{\perp}$ and $\Ll^{\perp}$ (see Lemma \ref{NM} and Definition \ref{kernels}). 
It follows that there are
matrix-valued functions $  A(z), B(z),  C(z), D(z)$ such that
\begin{align}\label{piri}
 P \Big ( i (1 - z) ( \Omega^* + \Gamma ) P_{\mathcal{N}}
  +   W(u_{-}^{1}, u_{+}^1) P_{\mathcal{N}^{\perp}}   
    + X(z) + Y(z)P_{\mathcal{N}^{\perp}}  \Big ) Q  = \begin{pmatrix}
  A(z)  & B(z) \\ C(z)  & D(z)
  \end{pmatrix}
\end{align}
{(recall the definitions of $X$ and $Y$ in Proposition \eqref{Tmain})} and 
\begin{align}\label{mainnnn}
A(z) &= i (1 - z) \boldsymbol{ A} + o(|1-z|),  \hspace{.3cm} B(z) = o(1),  \\ \nonumber \vspace{.5cm} C(z) &= O(|1-z|),   \hspace{.3cm} D(z) = \boldsymbol{ D} + o(1), 
\end{align}
where
\begin{align}
\boldsymbol{A} := [P_{\Ll}( \Omega^* + \Gamma )P_{\mathcal{N}}]_{\alpha}^{\beta}  ,  \hspace{1cm}  \boldsymbol{D} :=   [  P_{\Ll^\perp} W(u_{-}^{1}, u_{+}^1)P_{\mathcal{N}^\perp}]_{\gamma}^{\delta},   
\end{align}
 are invertible {where}  $[T]_{\theta}^{\eta}$ denotes {the matrix representation of a linear transformation $T$} in terms of the bases $\theta$, $\eta$. Here $\alpha=\{v_1,...,v_d\}$, which is a basis of $\mathcal{N}$, {$\beta=\{e_1,...,e_d\}$}, which is a basis of $\Ll$, and $\gamma=\{v_{d+1},...,v_{L}\}$ and {$\delta=\{e_{d+1},...,e_{L}\}$}, which are bases of $\mathcal{N}^\perp$ and $\Ll^\perp$, respectively.


\begin{theo}\label{T}
The transmission coefficients satisfy the following properties:
\begin{align}\label{899}
T^z_{+} =   T^1_{+}   + o(1),  \hspace{1cm}  T^z_{-} = T^1_{-}  + o(1){,}  
\end{align}
as $z$ tends to $1$ in $\overline{\mathbb{D}}$, where
\begin{align}\label{900}
T^1_{+} :=    Q  \begin{pmatrix}
  2   \boldsymbol{A}^{-1} & 0 \\ 0 & 0
\end{pmatrix}  P{,}   \hspace{1cm}  T^1_{-} :=  \Big (  
 Q  \begin{pmatrix}
  2   \boldsymbol{A}^{-1} & 0 \\ 0 & 0
\end{pmatrix}   P  \Big )^*.  
\end{align}
Moreover, 
\begin{align}\label{901}
  T_{+}^1 \mathbb{C}^L & = \mathcal{N},  \hspace{1cm} {\Ker(  T_{+}^1)}  = W( u_{-}^1, u_{+}^1 ) \mathbb{C}^L, \\ \notag    T_{-}^1 \mathbb{C}^L & = \Ll,  \hspace{1cm} {\Ker( T_{-}^1)}  = W( u_{+}^1, u_{-}^1 ) \mathbb{C}^L.  
\end{align}
The reflection coefficients satisfy the following properties:
\begin{align}\label{902}
R^z_{+} =   R^1_{+}   + o(1),  \hspace{1cm}  R^z_{-} = R^1_{-}  + o(1),  
\end{align}
as $z$ tends to $1$ in $\mathbb{S}^1$, where
\begin{align}\label{903}
R^1_{+} :=  \one - \Gamma  T^1_{+} ,  \hspace{1cm}  R^1_{-} := 
\one - \Gamma^{-1} T^1_{-}  + o(1).  
\end{align}
Moreover, 
\begin{align}\label{904}
  \Ll & = ( \one - R_{+}^1  ) \mathbb{C}^L   ,  \hspace{1cm}  {\Ker( T_{+}^1)}  = \Ker (  \one - R_{+}^1    ), \\ \notag    \mathcal{N} & =  ( \one - R_{-}^1  ) \mathbb{C}^L  ,  \hspace{1cm} {\Ker (T_{-}^1)}  = \Ker ( \one -  R_{-}^1    ) .  
\end{align}
\end{theo}

\begin{proof}
Equation \eqref{mainnnn} implies that the matrices $D(z)$ and $A(z) - B(z)D(z)^{-1} C(z)$ are invertible for $z\in \overline{\mathbb{D}}$  in a neighborhood of $1$, so  using the Schur complement formula, it follows that (recall \eqref{eq-NuDef})
\begin{align}\label{ccc}
&\notag  (\nu^{z})^{-1} \begin{pmatrix}
  A(z)  & B(z) \\ C(z)  & D(z)
  \end{pmatrix}^{-1} \\ &  =  (\nu^{z})^{-1} \begin{pmatrix}
  \one  & 0 \\ -D(z)^{-1} C(z) & \one 
  \end{pmatrix} \begin{pmatrix}
  \Big ( A(z) - B(z)D(z)^{-1} C(z)  \Big )^{-1}  & \!\!\!\!\! \!\!\! 0 \\ 0 &  \!\!\!\!\! \!\!\!D(z)^{-1} 
  \end{pmatrix} \begin{pmatrix}
  \one  & - B(z)D(z)^{-1} \\ 0 & \one 
  \end{pmatrix} \notag  \\
  & = \begin{pmatrix}
  2   \boldsymbol{A}^{-1} & 0 \\ 0 & 0
\end{pmatrix}   + o(1),  
\end{align}
where  \eqref{eq-NuDef} and  \eqref{mainnnn} were used.  The first equation in \eqref{900} follows from {Equation} \eqref{ccc}, Proposition  \ref{Tmain}, the continuity of $u^z_+(1)$,  \eqref{eq-MNId0}
and  Definition \ref{def-Scat}. The second equation in \eqref{900} is a consequence of the {first} equation, Definition \ref{def-Scat} and \eqref{eq-MNId2}.

\vspace{.1cm}

Due to the definition \eqref{change} of $P$, the kernel of $  T_+^1  $ is generated by  $  \{e_{d+1} , \cdots, e_L \}$, and they are a basis of $W(u_-^1, u_+^1) \mathbb{C}^L$: since $P$ is unitary, $P^* = P^{-1} =  \begin{pmatrix}
e_1  \cdots e_L
\end{pmatrix}  $. Therefore, the kernel at stake equals the kernel of   $ \begin{pmatrix}
\one & 0 \\ 0 & 0
\end{pmatrix}  \begin{pmatrix}
e_1  \cdots e_L
\end{pmatrix}^{-1}   $.  Due to the definition {\eqref{change} of $Q$,} the image of $  T_+^1  $ is generated by  $  \{v_1 , \cdots, v_d \}$, and these vectors form a basis of 
$ \mathcal{N} $. This proves the first line in \eqref{901}. The second one follows
from the fact that $T_{-}^1 = (T_{+}^1)^* $ (which can be deduced  from \eqref{eq-MNId2} and \eqref{transref}). 

\vspace{.1cm}

Next let us take  small enough $n$ such that $(u_-^{1}(n))^{-1}$ exists. Then, by continuity,   $(u_-^{1/z}(n))^{-1}$  exists, for $z$ in a neighborhood of  $1$. Using    
  \eqref{tresfromH} leads to
\begin{align}\label{si}
R^z_{+}= (u_-^{\overz}(n))^{-1}  u_-^z(n) -    
  ( u_-^{\overz}(n))^{-1} u_+^z(n) T^z_+ \to (u_-^{1}(n))^{-1}  u_-^1(n) -    
  ( u_-^{1}(n))^{-1} u_+^1(n) T^1_+ ,
  \end{align}
as $z $ tends to $ 1 $.   Taking the limit $n \to - \infty$  in the right hand side of \eqref{si}, we arrive at  the first equation in \eqref{903} (see also Lemma  \ref{defGamma} and \eqref{wasu}).  The second equation is obtained similarly.  {Equations} \eqref{904} follow from \eqref{901}, \eqref{903} and the fact that $\Gamma $ is a bijection from $\mathcal{N}$ onto 
$\Ll$.  
\end{proof}

\appendix

\section{Appendix}

{Let $\Mm={\mathcal{M}_{L\times L}}$ and $l^p(\mathbb{N},\Mm)$ the space of sequences $h:\mathbb{N} \to \Mm$ with $\sum_{n}\|h(n)\|^p<\infty$  for $p\in[1,\infty)$ and $\sup_n\|h(n)\|<\infty$ for $p=\infty$.}

\begin{theo}[{Lemma}~7.8 \cite{Tes}, Volterra Equation]\label{volterra} 
{Let  $g \in l^{\infty}(\mathbb{N},\Mm)$ and $K(n,m) \in \Mm$ for each $m,n\in \mathbb{N}$. 
Consider the Volterra sum equation}
\begin{align} \label{volt}
 f(n)=g(n)+\sum_{m=n+1}^{\infty}K(n,m)f(m),
 \end{align}
and suppose there is a sequence $M\in l^{1}(\mathbb{N},\mathbb{R})$  such that $\|K(n,m)\|\leq M(m)$ for each $m, n\in \mathbb{N}$. Then, {Equation} \eqref{volt} has a unique solution $f\in l^\infty(\mathbb{N},\Mm)$. Moreover, if $g(n)$ and $K(n,m)$ depend  continuously (resp. holomorphically) on a parameter $z$ (for every $n$),  $M$ does not depend on $z$, and $g(n)$ is uniformly bounded with respect to $n$ and $z$, then the same is true for $f(n)$.
\end{theo}
\begin{proof}
	For each $k\in \mathbb{N}$, if one finds a solution $f\in l^\infty(\mathbb{N}\cap[k,\infty),\Mm)$, then it can be extended to a solution in $l^\infty(\mathbb{N},\Mm)$ by defining recursively $f(n)=g(n)+\sum_{m=n+1}^{\infty}K(n,m)f(m)$ for each $n<k$. Since $M \in l^1(\mathbb{N},{\RM})$, there exists  $k\in \mathbb{N}$ such that $\sum_{m=k+1}^{\infty} M(m)<1/2$. Then, w.l.o.g., we can assume that $k=0$, {\it i.e.}, $\sum_{m=1}^{\infty}M(m)<1/2$. Then let us introduce the operator   $T:l^{\infty}(\mathbb{N},\Mm) \to l^{\infty}(\mathbb{N},\Mm)$  by $$(Tf)(n)=\sum_{m=n+1}^{\infty}K(n,m)f(m)$$ 
which is well-defined  because 
	$$\sum_{m=n+1}^{\infty}\|K(n,m)f(m)\|\leq\sum_{m=n+1}^{\infty}\|K(n,m)\|\|f(m)\|\leq 1/2 \|f\|_{\infty}. $$ Moreover, the last equation also implies that $T$ is bounded and $\|T\|<1/2$, therefore $I-T$ is invertible and $f:=(I-T)^{-1}g$ is a solution to the equation on $l^{\infty}(\mathbb{N},\Mm)$. 
	
\vspace{.1cm}

Now we assume that	$g(n) \equiv g^z(n)$ and $K(n,m) =  K^z(n,m)$ depend  continuously (resp. holomorphically) on a parameter $z$ (for every $n$),  $M$ does not depend on $z$, and $g^z(n)$ is uniformly bounded with respect to $n$ and $z$. Since the series  $(Tg^z)(n)=\sum_{m=n+1}^{\infty} K^z(n,m)g^z(m)$ converges uniformly,  $(Tg^z)(n)$ is then continuous (holomorphic) for each $n\in \mathbb{N}$ and it is uniformly bounded with respect to $n$ and $z$. Repeating the argument, one obtains that the same holds true for $ T^j g^z $, for every natural number $j$.   Using that $\|T^jg\|\leq (1/2)^j \sup_{n, z}\{\|g(n)\| \} $, it follows that the series  $f^z(n)=\sum_{j=0}^{\infty} (T^{j}g)^z(n)$ converges uniformly. This implies that the map $z\mapsto f^z(n)$ is continuous (holomorphic).
\end{proof}

\vspace{.1cm}

The following well-known result is recalled without proof {(see \cite{Gro} for a proof)}.

\begin{lemma}[Gronwall's Lemma]\label{gronwall}
	Let $(u_n)_{n\in \mathbb{N}}$, $(w_n)_{n\in \mathbb{N}}$ be {non}-negative real sequences and $\alpha$ a real number such that for $n\in \mathbb{N}$ 
	$$u_n\leq \alpha + \sum_{i=1}^{n-1}u_iw_i.$$
	Then, for $n\in \mathbb{N}$ we have that
	$$u_n\leq \alpha \exp\left(\sum_{i=1}^{n-1}w_i\right).$$
\end{lemma}

\begin{lemma}[Variation of parameters]\label{var_par}
	Consider the following difference equation
	\begin{equation}\label{diff}
		X(n-1)+AX(n)+X(n+1)=B(n)X(n),
	\end{equation}
	where $A , B(n) \in {\mathcal{M}}$ for each $n\in \mathbb{Z}$.
	Suppose that $S_1,S_2$ are solutions of the equation
	\begin{align}\label{PH5} 
	 X(n-1)+AX(n)+X(n+1)=0,
	\end{align}
	such that $S_1(0)=0$, $S_1(1)=\mathbf{1}$ and $S_2(0)=S_2(1)=\mathbf{1}$. Then, for $C,D \in {\mathcal{M}}$,  the solution $S$ to Equation \eqref{diff}, with initial conditions $S(0)=C$, $S(1)=D$, satisfies for $n\in \mathbb{N}$
\begin{align}\label{PH4}	
	S(n)=S_1(n)(D-C)+S_2(n)C+\sum_{j=1}^{n-1}S_1(n-j)B(j)S(j), 
\end{align}
	and for $n\in \mathbb{Z}^-\cup\{0\}$
	\begin{align} \label{como}
	S(n)=S_1(n)(D-C)+S_2(n)C-\sum_{j=n+1}^{0}S_1(n-j)B(j)S(j),
	\end{align}
where we identify $ \sum_{j=1}^{0}S_1(-j)B(j)S(j)  \equiv 0 $ and  $ \sum_{j=1}^{0}S_1(1-j)B(j)S(j) \equiv 0$.    
\end{lemma}
\begin{proof}
Equations \eqref{PH4} and \eqref{como} together with the initial conditions define recursively a matrix valued function that is denoted by $S$. Now we prove that $S$ satisfies {Equation}  \eqref{diff}.  The proof  is carried out only for $n \geq 2$, the other cases are similarly treated.  
For every $ X \in {\mathcal{M}^\ZM}  $,
we {use the notation}
$$h_n( X) =  X(n-1)+AX(n)+X(n+1) .$$ Using Equation~\eqref{PH4} and $n \geq 2$ we get (here recall that $S_1$ ans $S_2$ satisfy \eqref{PH5} and, therefore, $  h_n(S_1) = 0 =  h_n(S_2) $; moreover $S_1(0)= 0,   S_1(1)={\bf 1}$)
\begin{align*}
h_n(S)  = & \, h_n(S_1)(D-C) +  h_n(S_2)C + \sum_{j=1}^{n-2} h_{n-j} (S_1) B(j)S(j)
\\ \notag &
+  A S_1(n -(n-1))B(n-1)S(n-1) +  S_1(n +1 -(n-1))B(n-1)S(n-1)  \\ \notag & + 
 S_1(n +1 - n)B(n)S(n)
 \\ = & \,A  S_1(1) B(n-1)S(n-1) +   S_1(2)B(n-1)S(n-1) +  B(n)S(n)
 \\ = &\,\Big ( A S_1(1)   +   S_1(2) \Big )B(n-1)  S(n-1) +  B(n)S(n) 
 \\ = & \,-S_1(0)B(n-1) S(n-1) +  B(n)S(n) = B(n)S(n). 
\end{align*} 
We obtain that $h_n(S) =  B(n)S(n)$, which is \eqref{diff}. 
\end{proof}

\begin{lemma}\label{PPH}
{Consider the set-up of Theorem \ref{volterra} with $g=\mathbf{1}, K^z(n,j)=-z^{j-n}s^z(j-n)V (j)$ and $M(j)=j\|V(j)\|$ where $s^z$ is as in Definition~\ref{FreeSolo1}. Furthermore let $\tilde{u}_+^z$ be} the corresponding solution to the  Volterra equation 
(for $n \in \mathbb{N}$)
 \begin{align}\label{PH8}	\tilde{u}_{+}^z(n)=\mathbf{1} - \sum_{j=n+1}^{\infty} s^z(j-n)V(j)z^{j-n}\tilde{u}_{+}^z(j).
\end{align} 
{Finally} denote by   $u_+^z(n)=z^n\tilde{u}_+^z(n)$,  $n \in \mathbb{N}$. It follows that
\begin{align}\label{PH9}
 u_+^z(n-1)  +V(n)u^z_+(n) + u^z_+(n+1) =(z+1/z)u^z_+(n),   \hspace{2cm}  n \geq 2.
\end{align}
\end{lemma}

\begin{proof}
Let us take $n \geq 2$.  We use the notation of  Lemma \ref{var_par} and its proof, taking $A=-(z+1/z)\mathbf{1}$, and set $\gamma(n) = z^n$.  A direct calculation shows that
$h_m(s^z) = 0$ and $h_m (\gamma) =0 $, for every $m$. Further let us note that   
 \begin{align}\label{PH10} u_{+}^z(n)=  \gamma(n) + \sum_{j=n+1}^{\infty} s^z(n-j)V(j) u_{+}^z(j).
\end{align} 
It follows from \eqref{PH10} by an algebraic calculation using the definition of $h_m$ that 
 \begin{align}
 \label{PH1p1}
  h_n(u_{+}^z)= \,& h_n(\gamma) + \sum_{j=n+2}^{\infty} h_{n-j}(s^z)V(j) u_{+}^z(j) 
  + A s^z(n-(n+1)) V(n+1) u_{+}^z(n+1) \\ & \notag  +  
  s^z(n-1-(n+1)) V(n+1) u_{+}^z(n+1) 
  +s^z(n-1-n) V(n) u_{+}^z(n). 
\end{align} 
Using \eqref{PH1p1}, the facts that $h_m (s^z) = 0$,  $h_n(\gamma)  =0 $ and $s^z(m) = - s^z(-m)$ together with $s^z(0)= 0$,  $s^z(1) = 1$, one gets that 
 \begin{align}
\notag
  h_n(u_{+}^z)=\, &  -(A s^z(1) + s^z(2))   V(n+1) u_{+}^z(n+1) - s^z(1) V(n)   u_{+}^z(n)
\\ \notag   =\, & \, s^z(0)  V(n+1) u_{+}^z(n+1) -  V(n)   u_{+}^z(n) 
\\ \notag   = \,& -  V(n)   u_{+}^z(n), 
\end{align} 
which is 
\eqref{PH9}. 
\end{proof}

\noindent {\noindent {\bf Acknowledgements:} This research was supported by CONACYT, FORDECYT-PRONACES
	 429825/2020 (proyecto apoyado por el FORDECYT-PRONACES, PRONACES/429825). The work of M. B., G.G. and G.F. was also supported by the project PAPIIT-DGAPA-UNAM IN101621,  that of H. S.-B. alos by
	PAPIIT-UNAM IN105718, CONACYT Ciencia Basica 283531 and the DFG SCHU 1358/6-2.}


\end{document}